\documentclass[fontsize=11pt,a4paper,DIV10]{scrartcl}
\usepackage{graphicx,amssymb,amsmath,amsthm}

\usepackage{xspace}
\usepackage{subcaption}
\usepackage{microtype}
\usepackage{nicefrac}
\usepackage{xcolor}
\usepackage{thm-restate}
\usepackage{url}
\usepackage{paralist}

\usepackage{lineno}
\linenumbers
\nolinenumbers

\usepackage{hyperref}
\usepackage{cleveref}

\hypersetup{
	colorlinks   = true,
	citecolor    = blue,
	urlcolor=cyan
}

\graphicspath{{./figures/}}

\newcommand{\hex}{\ensuremath{C_6}\xspace}
\newcommand{\dhex}{\ensuremath{C_{10}}\xspace}
\newcommand{\Px}{\ensuremath{P_X}\xspace}
\newcommand{\Pu}{\ensuremath{P_U}\xspace}
\newcommand{\Pl}{\ensuremath{P_L}\xspace}
\newcommand{\Pz}{\ensuremath{P_Z}\xspace}
\newcommand{\Pstrich}{\ensuremath{P_{-}}\xspace}

\newcommand{\Oct}{\ensuremath{\mathcal{O}}\xspace}
\newcommand{\C}{\ensuremath{\mathcal{C}}\xspace}
\newcommand{\nC}{\ensuremath{\mathcal{\overline{C}}}\xspace}
\newcommand{\Cfree}{\C-free\xspace}
\newcommand{\Pfree}{\ensuremath{\mathcal P}-free\xspace}
\newcommand{\Tcontains}{$\triangle$-contains\xspace}
\newcommand{\Tcontain}{$\triangle$-contain\xspace}
\newcommand{\cover}{covering number\xspace}

\newcommand{\scaleL}{.58}

\newtheorem{theorem}{Theorem}
 \newtheorem{obs}[theorem]{Observation}

\newtheorem{lem}[theorem]{Lemma}

\newcounter{claimcounter}
\numberwithin{claimcounter}{theorem}
\newtheorem{claim}[claimcounter]{Claim}

\crefname{theorem}{Theorem}{Theorems}
\crefname{obs}{Observation}{Observations}
\crefname{lemma}{Lemma}{Lemmas}
\crefname{lem}{Lemma}{Lemmas}
\crefname{proposition}{Proposition}{Proposition s}
\crefname{section}{Section}{Sections}
\crefname{figure}{Figure}{Figures}
\crefname{claim}{Claim}{Claims}
\crefname{table}{Table}{Tables}
\crefname{corollary}{Corollary}{Corollaries}
\crefname{enumi}{}{}

\newcommand{\new}[1]{#1}

\begin{document}
	\title{Folding Polyiamonds into Octahedra}
	\author{
		\parbox{6.7cm}{\center
			{\textsc{Eva Stehr}}\\[3pt]
			\small
			\url{eva@stehr.dev}\\
			{TU Braunschweig}}
		\and
		\hspace{-2cm}
		\parbox{6.7cm}{\center
			{\textsc{Linda Kleist}}\\[3pt]
			\small
			\url{kleist@ibr.cs.tu-bs.de}\\
			{TU Braunschweig}}
	}

	\date{July 2022}
	
	\maketitle

\begin{abstract}
We study polyiamonds (polygons arising from the triangular grid) that fold into the smallest yet unstudied platonic solid -- the octahedron.
We show a number of results. Firstly, we characterize foldable polyiamonds containing a hole of positive area, namely each but one polyiamond is foldable. Secondly, we show that a convex polyiamond folds into the octahedron if and only if it contains one of five polyiamonds.
We thirdly present a sharp size bound:
While there exist unfoldable polyiamonds of size 14, every polyiamond of size at least 15 folds into the octahedron. \new{This clearly implies that one can test in polynomial time whether a given polyiamond folds into the octahedron.} Lastly, we show that for any assignment of positive integers to the faces, there exist a  polyiamond that folds into the octahedron such that the number of triangles covering a face is equal to the assigned number.
\end{abstract}

\section{Introduction}
Algorithmic origami is a comparatively young branch of computer science that studies the algorithmic aspects of folding various materials. The construction of three-dimensional objects from two-dimensional raw materials is of particular interest and has applications in robotics in general \cite{origami-stent-grafts,sheet-metal-bending},  and also in the construction of objects in space~\cite{programmable}.

\begin{figure}[htb]
	\centering
	\begin{subfigure}[t]{.3\textwidth}
		\centering
		\includegraphics[page=5,scale=.6]{gridObs}
		\caption{A polyiamond $P$.}
		\label{fig:polyiamond}
	\end{subfigure}\hfil
	\begin{subfigure}[t]{.3\textwidth}
		\centering
		\includegraphics{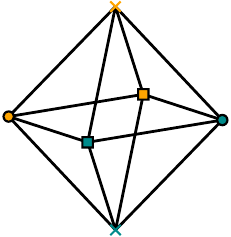}
		\caption{The octahedron \Oct.}
		\label{fig:octa}
	\end{subfigure}\hfil
	\caption{Does the polyiamond $P$ fold into the octahedron \Oct?}
	\label{fig:ex1}
\end{figure}

While foldings of 
polycubes and tetrahedra have already been studied, we take the next step and focus on the question of whether a given polyiamond folds into the octahedron, e.g., does the polyiamond in \cref{fig:ex1} fold into the octahedron?

\paragraph{Terminology}\label{sec:prelim}
By the \emph{octahedron}~\Oct, we refer to the regular octahedron composed of eight equilateral (unit) triangles; for an illustration consider \cref{fig:octa}. Note that four triangles meet in each of the six corners of the octahedron.
Because all faces of the octahedron are triangles, our pieces of paper are polygons arising from the triangular grid. 
A \emph{polyiamond} of size $n$ is a connected polygon in the plane formed by joining $n$ triangles from the triangular grid by identifying some of their common sides; for an example consider \cref{fig:polyiamond}.
To avoid confusion with the corners of the octahedron, we refer to the vertices of the triangles forming $P$ as the \emph{vertices} of $P$; note that these vertices may also lie inside~$P$.

We view $P$ as a set which includes the $n$ open triangles and a subset of the unit-length boundary edges \new{shared by any two of these triangles}; the existence of such an edge models the fact that the \new{two incident} triangles are glued along this side. 
Because we only want robust connections between triangles via their sides, we do not specify the existence or non-existence of vertices which do not influence the foldability.
However, for the upcoming definitions of slits and holes, we assume that the vertices do not belong to the polyiamond. 

If a shared edge does not belong to $P$, we call it a \emph{slit edge}.
We also allow the polyiamonds to have holes; 
a 
\emph{hole} of a polyiamond
  is a bounded connected component of its complement,
  which is different from a single vertex.  
We call a hole a \emph{slit} if it has area zero and consists of one or more slit edges. 
We consider two polyiamonds to be  \emph{the same} if they are congruent, i.e., if they can be transformed into one another by a set of translations, rotations and reflections. Moreover, a polyiamond is \emph{convex} if it forms a convex set in the plane (after adding a finite set of points corresponding to vertices).

\paragraph{Folding model}

We consider foldings in the \emph{grid folding model}, where
 folds along the grid lines are allowed such that in the final state every triangle covers a face of the octahedron, i.e., we forbid folding material strictly outside or inside the octahedron. Consequently, in the final state the folding angles are $\pm\beta:=\arccos(\nicefrac{1}{3}) $ or $\pm 180^\circ$. 
Moreover, a folding of a polyiamond $P$  into the octahedron \Oct  induces a 
\emph{triangle-face-map}, i.e., a mapping of the triangles of~$P$ to the faces of \Oct.
We say  \emph{$P$ folds into \Oct} (or \emph{$P$ is foldable}), if $P$ can be transformed by folds along the grid lines into a folded state such that the induced triangle-face-map is surjective, i.e.,  each face of \Oct is covered by at least one triangle. 
 In order to study non-foldable polyiamonds, we also consider partial foldings, i.e., foldings where potentially not all faces of \Oct are covered. Note that partial foldings induce triangle-face-maps that are not necessarily surjective.

\subsection{Related work} 
Past research has particularly focused on
 folding polyominoes into \emph{polycubes}. Allowing for folds along the \emph{box-pleat grid} (consisting of square grid lines and alternating diagonals), 
 Benbernou et  al.\  (with differing co-authors) show that every polycube~$Q$  of size $n$
 can be folded from a sufficiently large 
 square polyomino~\cite{universalHingePatternsOrthoShapes} or from a $2n\times 1$ strip-like polyomino~\cite{universalHingePatternsGridPolyhedron}. 
 Moreover, common unfoldings of polycubes have been investigated in the grid model. 
 The \emph{(square) grid model} allows folds along the grid lines of a polyomino with fold angles of $\pm90^\circ$ and $\pm 180^\circ$, and allows material only on the faces of the polyhedron.
 Benbernou et  al.\ show that there exist polyominoes that fold into all polycubes with bounded surface area~\cite{universalHingePatternsGridPolyhedron} and Aloupis et  al.\  study 
 common unfoldings of various classes of polycubes~\cite{unfoldings}.
 Moreover, there exist polyominoes that fold into several different boxes~\cite{rec0,rec1,rec2,rec3,rec4}.

Decision questions for folding (unit) cubes are studied by Aichholzer et  al.~\cite{foldingCube_withHoles_Journal,foldingCubes}. The \emph{half-grid model} allows folds of all degrees along the grid lines, the diagonals, as well as along the horizontal and vertical halving lines of the squares. In this model, every polyomino of size at least 10 folds into the cube~\cite{foldingCubes}. 
The remaining polyominoes of smaller size are explored by Czajkowski et  al.~\cite{2020smallFoldingCube}.
In the grid model, Aichholzer et al.~\cite{foldingCubes} characterized the foldable tree-shaped polyominoes that fit within a \(3 \times n\) strip. Investigating polyominoes with holes, 
Aichholzer et al. \cite{foldingCube_withHoles_Journal}  show that all but five \emph{basic} holes (a single unit square, a slit of length~1, a straight slit of length~2, a corner slit of length 2 and a U-shaped slit of length~3) guarantee that the polyomino folds in the grid model into the cube.

In the context of polyiamonds, Aichholzer et al.~\cite{foldingCubes} present a nice and simple characterization of polyiamonds that fold into the smallest platonic solid: Even when restricting to folds along the grid lines, a polyiamond folds into the tetrahedron if and only if it contains one of the two tetrahedral nets.

\paragraph{Results and organization} 
In this work, we study foldings to the smallest yet unstudied platonic solid -- the octahedron.
Our main results are as follows:
\begin{compactitem}
	\item 
	In \cref{sec:2},  we identify some sufficient and necessary conditions for foldability  and take a closer look at polyiamonds with slits and holes. 
	\item 
	Among our findings in \cref{sec:2new}, we characterize foldable polyiamonds containing a hole of positive area: each but one polyiamond is foldable.
	\item In \cref{sec:3}, we characterize the convex foldable polyiamonds: A convex polyiamond folds into \Oct if and only if it contains one of five polyiamonds.
	\item In \cref{sec:4}, we show that every polyiamond of size $\geq 15$ is foldable. A non-foldable polyiamond of size 14 proves that this bound is best possible. {We highlight that an analogous statement for folding polyominoes into cubes does not exist in the grid folding model, i.e., there exist arbitrarily large polyominoes that do not fold into the cube. For instance, large rectangular \new{polyominoes \cite[Corollary 2]{foldingCube_withHoles_Journal}}.}
	\item In \cref{sec:5}, we discuss the reverse question: For an assignment of positive integers to the faces of the octahedron \Oct, does there exist a polyiamond that folds into \Oct such that the number of triangles covering each face is equal to the assigned number? We prove that such a polyiamond exists for any such assignment.
\end{compactitem}

\section{Some Tools}\label{sec:2}

In this section, we present tools for proving or disproving the foldability of a polyiamond into an octahedron.
Firstly, we present a useful connection between 3-colorings of the triangular grid and polyiamonds inheriting the vertex coloring that are folded into the octahedron.
As indicated in \cref{fig:grid}, the triangular grid graph allows for a proper 3-coloring. 
Because every (connected) inner triangulation has at most one 3-coloring (up to exchange of the colors), every polyiamond has a unique 3-coloring which is induced by the triangular grid. If there exists a slit edge along a grid line, the polyiamond graph may have several vertices corresponding to one grid vertex, see also \cref{fig:slitsB}.
Note that each corner of \Oct has a unique non-adjacent corner which we call its \emph{antipodal}. 

\begin{figure}[hb]
	\centering
	\includegraphics[scale=1]{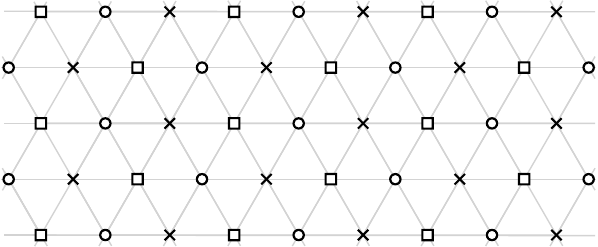}
	\caption{A 3-coloring of the triangular grid.}
	\label{fig:grid}
\end{figure} 

In order to study the non-foldability, we also consider partial foldings. In particular, when relaxing the condition that all faces are covered, we say a polyiamond is \emph{partially} folded into the octahedron.

\begin{restatable}{lem}{lemcolorclasses}\label{lem:colorclasses}
	Let $P$ be a polyiamond with a 3-coloring of its vertices.
	In every (partial) folding of~$P$ to the octahedron~\Oct, the vertices of each color class are mapped to (one corner or a pair of) antipodal corners of \Oct.
\end{restatable}

\begin{proof}
	Consider two neighboring triangles of $P$ and note that their two private vertices have the same color. If their common side is folded by $\pm\beta$, these two vertices are mapped to antipodal corners of \Oct; otherwise the edge is folded by $\pm 180^\circ$ and the two vertices are mapped to the same corner of $\Oct$. The fact that $P$ is connected implies that every color class is mapped to a different set of antipodal corners.
\end{proof}

We repeatedly use  \cref{lem:colorclasses} in order to disprove foldability. Moreover, for folded polyiamonds, \cref{lem:colorclasses}  allows to illustrate the mapping of vertices to corners of \Oct by vertex colorings where antipodal corners of \Oct are represented by the same shape but different colors, for an example consider  \cref{fig:nets}. Clearly, such a vertex coloring induces a triangle-face map.
	
\subsection{Foldability}
A polyiamond~$P$ \emph{contains} a polyiamond~$P'$ if $P'$ can be translated, rotated, and reflected such that all triangles and triangle sides of $P'$ also belong to $P$. Restricting our attention to the triangles, a polyiamond~$P$ \emph{\Tcontains} a polyiamond $P'$ if all triangles of $P'$ belong to $P$. 
 For example, the polyiamond in \cref{fig:slitsB}  does not contain  but \Tcontains the polyiamond  in \cref{fig:slitsA}.
As we will see  in \cref{obs:containment},  neither containment nor $\triangle$-containment of a foldable polyiamond is a sufficient folding criterion.
Nevertheless, we are able to show two sufficient  criteria based on $\triangle$-containment of foldable polyiamonds. 
By zig-zag-folding as indicated in \cref{fig:convexReduction}, every polyiamond can be reduced to a contained convex polyiamond.

\begin{restatable}{lem}{lemcontainingConvex}\label{lem:containingConvex}
	A polyiamond $P$ is foldable if it \Tcontains a convex foldable polyiamond $C$.
\end{restatable}

\begin{proof}
	Firstly, we reduce $P$ to $C$: For every boundary side $s$ of $C$, we fold the triangles of $P$ outside $C$ in a zig-zag-manner. To this end, we fold along the grid lines parallel to $s$ with $+180^\circ$ and $-180^\circ$ folds alternatingly,  as illustrated in \cref{fig:convexReduction}. 

	\begin{figure}[htb]
		\centering
		\includegraphics[page=4]{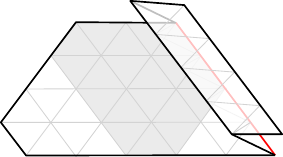}
		\caption{Folding strategy to reduce a polyiamond to a convex subpolyiamond by zig-zag-folding the outside.}
		\label{fig:convexReduction}
	\end{figure} 
	As a result,  the supporting line of $s$ bounds the folded polyiamond.
	Because $C$ is convex and contained in $P$, $P$ can be transformed to $C$ with the above procedure. 
	Secondly, we use the fact that $C$ folds into \Oct.
\end{proof}

 A \emph{net} of 
a  polyhedron is formed by cutting along certain edges and unfolding the resulting connected set to lie flat \new{without intersections}. Nets of the octahedron are  depicted in \cref{fig:nets}.
\begin{figure}[htb]
	\centering
	\includegraphics[page=6]{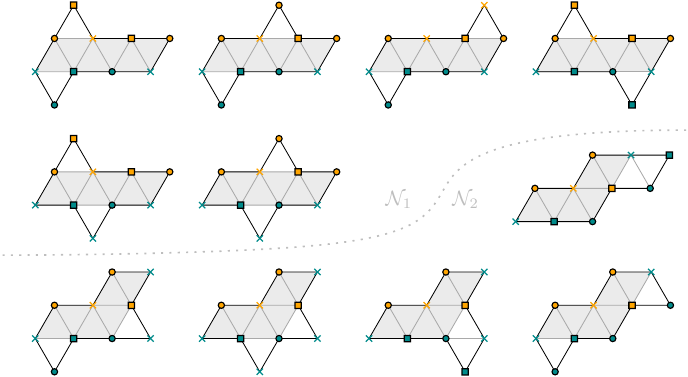}
	\caption{The eleven nets of the octahedron split into two groups $\mathcal N_1$ and $\mathcal N_2$.}
	\label{fig:nets}
\end{figure} 
There exist two interesting facts for nets of three-dimensional regular convex polyhedra~\cite{Buekenhout1998}:  
Firstly, each net  is uniquely determined by a spanning tree of the 1-skeleton of the
polyhedron, i.e.,  the cut edges form a spanning tree of the vertex-edge graph. Secondly, dual polyhedra (e.g., the cube and the octahedron) have the same number of nets. Consequently, there exist eleven octahedron nets. 

We show that \Tcontain{ing} a net is a sufficient folding criterion for a polyiamond.

\begin{restatable}{lem}{lemnets}\label{lem:nets}
	A polyiamond is foldable if it \Tcontains an octahedron net.
\end{restatable}

\begin{proof}
	We partition the set of nets into two groups $\mathcal N_1$ and  $\mathcal N_2$ as illustrated in \cref{fig:nets}. Note that within each group, the vertex-corner-maps (can be shifted such that they) are consistent on common triangles. For $i=1,2$, we consider the smallest convex polyiamond $S_i$ containing the nets of $\mathcal N_i$ as depicted in \cref{fig:netsProof}.  The coloring of the vertices gives a mapping to the corners of \Oct
		 and thus describes a folding of $S_i$ into $\Oct$.
	
	\begin{figure}[htb]
		\centering
		\begin{subfigure}[t]{.45\textwidth}
				\centering
				\includegraphics[page=11]{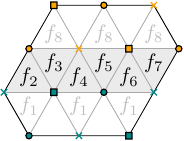}
				\caption{$S_1$}
				\label{fig:S1}
		\end{subfigure}\hfill
			\begin{subfigure}[t]{.45\textwidth}
		\centering
		\includegraphics[page=12]{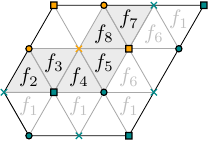}
		\caption{$S_2$}
		\label{fig:S2}
	\end{subfigure}
		\caption{Illustration for the proof of \cref{lem:nets}.}
		\label{fig:netsProof}
	\end{figure} 
	
	By \cref{lem:containingConvex}, all polyiamonds that contain a convex polyiamond can be reduced to the convex polyiamond. By construction, not all triangles of $S_i$ are present in each net of~$\mathcal N_i$. However, each net has at least eight triangles with pairwise different labels.
	The non-existence of a triangle harms the foldability only if it is essential to cover a face.
\end{proof}

\subsection{Non-foldability}
The following lemma is a crucial tool to disprove foldability. 
To this end, let \hex and \dhex denote the polyiamonds depicted in \cref{fig:C6a,fig:C10b}, respectively.

\begin{figure}[htb]
	\centering
		\begin{subfigure}[t]{.3\textwidth}
		\centering
		\includegraphics[page=4]{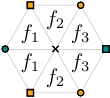}
		\caption{A triangle-face-map of \hex covering three faces of \Oct.}
		\label{fig:C6a}
	\end{subfigure}\hfill
		\begin{subfigure}[t]{.3\textwidth}
			\centering
			\includegraphics[page=3]{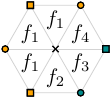}
			\caption{A triangle-face-map of \hex covering four faces of \Oct.}
			\label{fig:C6b}
		\end{subfigure}\hfill
		\begin{subfigure}[t]{.33\textwidth}
			\centering
			\includegraphics[page=16]{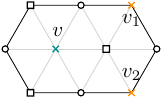}
			\caption{Any triangle-face-map of \dhex covers  at most six faces of~\Oct.}
			\label{fig:C10b}
		\end{subfigure}
	\caption{Illustration of \cref{lem:hex} and its proof.}
	\label{fig:C6}
\end{figure} 

\begin{restatable}{lem}{lemhex}\label{lem:hex}
	Let $P$ be a polyiamond (partially) folded into the octahedron \Oct.
	\begin{compactenum}[(i)]
		\item \label{prop:hex} Every \hex that is \Tcontain{ed} in $P$ covers at most four different faces of \Oct.
		\item \label{prop:hex2}  If a \hex in $P$ covers exactly three or four faces, then the induced triangle-face-mapping is unique (up to symmetry) and  as depicted in \cref{fig:C6a,fig:C6b}, respectively.
		\item \label{prop:dhex}  Every \dhex contained in $P$ covers at most six different faces of \Oct.
	\end{compactenum}
\end{restatable}

\begin{proof}
	Let $v$ denote the central vertex of \hex. In the folded state, $v$ is mapped to a corner $c$ of the octahedron which is (like every corner) incident to four faces.
	\begin{compactenum}[(i)]
		\item  Because every triangle of \hex is incident to vertex $v$, these triangles cover a subset of the four faces incident to $c$.
		\item 
		We consider a 3-coloring as indicated in \cref{fig:C6a}.
		If all circle or all square vertices  map to a same corner, then \hex covers at most two faces of~\Oct, namely the ones incident to the cross and circle vertex. Hence, if $P$ covers three or four faces, then exactly two vertices of each class map to the same corner and the third vertex of each class maps to its antipodal, we call this vertex \emph{lonely}. We distinguish whether the two lonely vertices are a) adjacent or b) opposite in \hex, see \cref{fig:C6a,fig:C6b}. It follows that the number of covered faces is three and four, respectively.
		
		\item 
		We consider a 3-coloring of \dhex as illustrated in \cref{fig:C10b} and use the fact that each color class is mapped to antipodal corners by \cref{lem:colorclasses}. We denote the three cross vertices by $v, v_1, v_2$ as illustrated in \cref{fig:C10b}; similarly, we denote the corner of \Oct to which $v$ is mapped by $c$.  If at most one $v_i$ (which are both incident to only two triangles) is mapped to the antipodal corner $\overline c$ of $c$, then at most two faces incident to $\overline c$ can be covered.
		If both $v_1$ and $v_2$ are mapped to $\overline c$, then the four incident triangles of $\overline c$ share a common edge. Consequently, they may cover at most two incident faces. 
		In other words, in both cases at least two faces (incident to $\overline c$) remain uncovered and thus a \dhex covers at most six faces.
	\end{compactenum}
This completes the proof.
\end{proof}

\section{On Slits and Holes}\label{sec:2new}
In this section, we consider polyiamonds with slits and holes. First of all,  we remark 
 that removing individual edges from a foldable polyiamond does not destroy its foldability as long as  connectivity is maintained. This allows us to focus on polyiamonds without slit edges, i.e.,  \emph{sealing} slit edges may only increase the level of difficulty to prove foldability.
On the other hand, we note that slits may in fact enable foldability.

\begin{obs}\label{obs:containment}
	Let  $P$ be a polyiamond 	($\triangle$-)containing a foldable polyiamond $P'$. Then, the polyiamond $P$ may not be foldable.
\end{obs}

As we show in \cref{prop:unfoldable12} in the proof of \cref{thm:convex}, the polyiamond $P$ in \cref{fig:slitsA} 
\begin{figure}[thb]
	\centering
	\begin{subfigure}[t]{.35\textwidth}
		\centering
		\includegraphics[page=2]{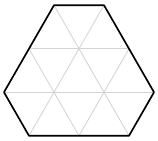}
		\caption{By \cref{thm:convex}, the depicted polyiamond does not fold into~\Oct.
		}
		\label{fig:slitsA}
	\end{subfigure}\hfill
	\begin{subfigure}[t]{.55\textwidth}
		\centering
		\includegraphics[page=15]{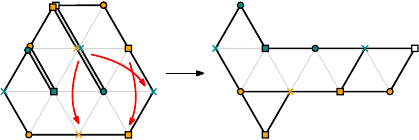}
		\caption{With additional slit edges, the polyiamond folds into \Oct. }
		\label{fig:slitsB}
	\end{subfigure}\hfil
	\caption{Illustration for \cref{obs:containment}.}
	\label{fig:slits}
\end{figure}
does not fold into~$\Oct$, while the polyiamond $P'$ with additional slit edges in \cref{fig:slitsB}
 can be transformed into a polyiamond \Tcontain{ing} a net. Hence, $P'$ is foldable by \cref{lem:nets}.

We now characterize foldable polyiamonds with holes of positive area. Let $O$ denote the polyiamond illustrated in \cref{fig:hole0}.

\begin{restatable}{theorem}{thmHoles}\label{thm:Holes}
	Let $P$ be a polyiamond containing a hole~$h$ of positive area. Then $P$ folds into~\Oct if and only if it is not the polyiamond $O$.
\end{restatable}

\begin{proof}
	The non-foldability of $O$ is analogous to the proof that the polyiamond depicted in \cref{fig:slitsA} is non-foldable, see \cref{prop:unfoldable12} below.

	For the reverse direction, we focus on a largest hole $h$ with positive area and distinguish two cases:

	\begin{figure}[b]
		\centering
		\begin{subfigure}[t]{.25\textwidth}
			\centering
			\includegraphics[page=8]{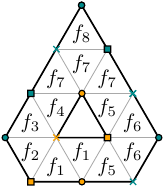}
			\caption{Polyiamond $P_a$}
			\label{fig:holeA}
		\end{subfigure}\hfil
		\begin{subfigure}[t]{.25\textwidth}
			\centering
			\includegraphics[page=5]{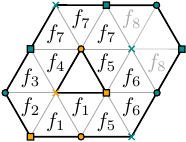}
			\caption{Polyiamond $P_b$}
			\label{fig:holeB}
		\end{subfigure}\hfil
		\begin{subfigure}[t]{.25\textwidth}
			\centering
			\includegraphics[page=7]{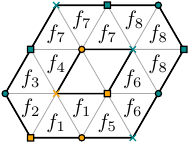}
			\caption{Polyiamond $P_c$}
			\label{fig:holeC}
		\end{subfigure}\hfil
		\begin{subfigure}[t]{.2\textwidth}
		\centering
		\includegraphics[page=9]{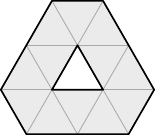}
		\caption{Polyiamond $O$}
		\label{fig:hole0}
	\end{subfigure}\hfil
		
		\caption{llustration for the proof of \cref{thm:Holes}.}
		\label{fig:hole}
	\end{figure} 

	If $h$ contains two neighboring triangles, then we reduce $P$ to the polyiamond~$P_c$ depicted in \cref{fig:holeC} as follows: we choose two neighboring triangles of $h$ which form a (potentially smaller) hole $h'$ in the form of a parallelogram. Then, we fold all triangles that do not touch $h'$ with a vertex or edge by zig-zag-folding the outside as in \cref{fig:convexReduction}. This results in the polyiamond~$P_c$ because $h$ and thus $h'$ are enclosed by a cycle of triangles of $P$. Moreover, it is easy to check that $P_c$ is foldable, e.g., when inducing the triangle-face-map depicted in \cref{fig:holeC}.
	
	It remains to consider the case that $h$ contains \new{exactly one} triangle and $P$ is not $O$. 
	If $P$ can be reduced (by zig-zag-folding) to the polyiamond $P_a$ depicted in \cref{fig:holeA}, then $P$ folds into~\Oct \new{by \cref {lem:containingConvex}}. 
	Otherwise, we use zig-zag-folds to obtain a subpolyiamond $P'$ of $P_b$ depicted in \cref{fig:holeB}. Because $P$ is different from $O$ and cannot be reduced to $P_a$, this ensures that $P'$  has at least one triangle with label~$f_8$. Because $P_b$ folds into \Oct, so does $P$. 
\end{proof}


\section{Characterization for Convex Polyiamonds}\label{sec:3}
In this section, we characterize convex foldable polyiamonds.
Let $\mathcal C $ denote the set of five convex polyiamonds depicted in \cref{fig:convex_foldable}.

\begin{figure}[htb]
	\centering
	\includegraphics[page=12]{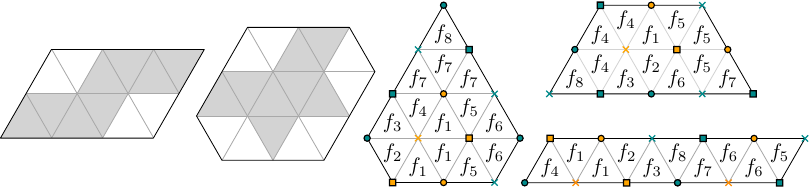}
	\caption{Illustration for \cref{thm:convex}; the set \C of foldable polyiamonds and their foldings. }
	\label{fig:convex_foldable}
\end{figure}

\begin{restatable}{theorem}{thmConvex}\label{thm:convex}
	A convex polyiamond $P$ folds into \Oct if and only if it contains one of the five polyiamonds in \C.
\end{restatable}

\begin{proof}
	First, we show that a convex polyiamond $P$ folds into \Oct if it ($\triangle$-)contains a polyiamond in~\C. Note that each polyiamond in  \C is convex. Hence, by \cref{lem:containingConvex}, it suffices to present folding strategies for the polyiamonds in $\C$, see \cref{fig:convex_foldable}. While two polyiamonds contain an octahedral net, we present explicit strategies for the remaining three.

	Second, we show that every convex polyiamond that folds into \Oct contains a polyiamond from \C. 
	To do so, we construct all convex \emph{\Cfree} polyiamonds, i.e., all convex polyiamonds that contain none of the five polyiamonds in \C. The construction is as follows, for an illustration
	consider \cref{fig:convex_foldable_constructionTREE}:
	
	\begin{figure}[p]
		\centering
		\includegraphics[page=7]{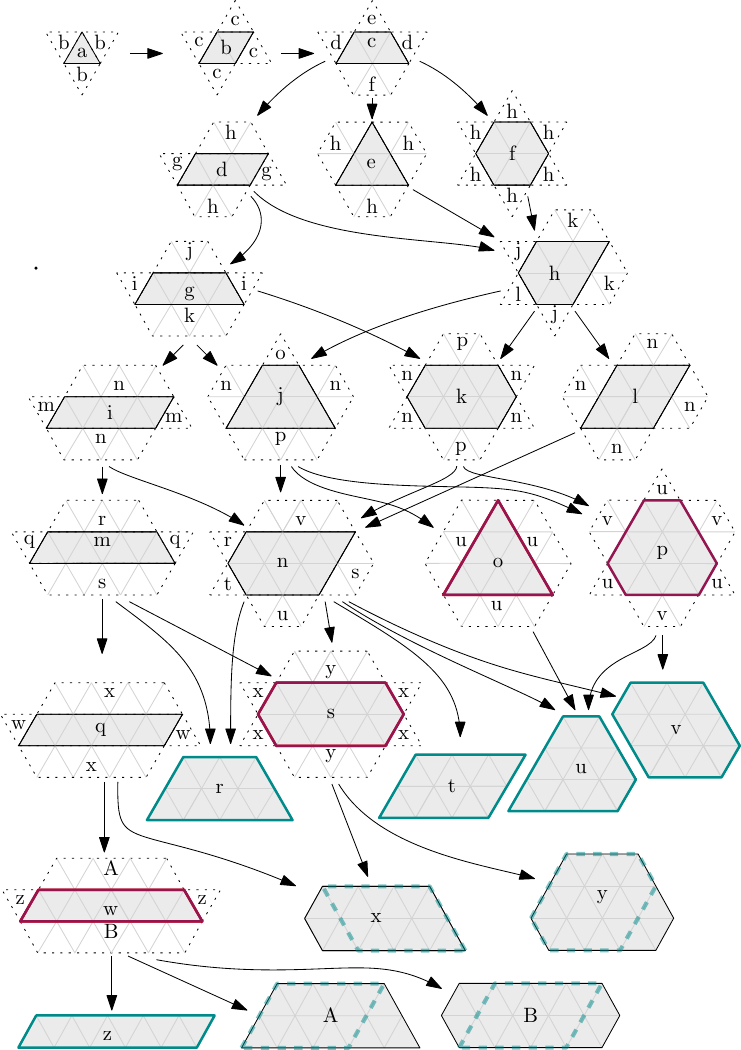}
		\caption{Construction of all \Cfree polyiamonds;
			the  inclusion-wise maximal \Cfree polyiamonds $o$, $p$, $s$, and $w$  are highlighted in red. 
		}
		\label{fig:convex_foldable_constructionTREE}
	\end{figure}

	We start with the unique polyiamond of size 1. Then, we consider all possibilities to enlarge every constructed polyiamond by one triangle and extend it to the smallest convex polyiamond contain{ing} it, i.e.,  we add just enough triangles such that the resulting polyiamond is convex again. We stop when we encounter a polyiamond from \C or a polyiamond contain{ing} one of them.
	
	By their convexity and \cref{lem:containingConvex}, it suffices to show the non-foldability of the inclusion-wise maximal \Cfree polyiamonds. The construction shows that the set~\nC of inclusion-wise maximal \Cfree polyiamonds consists of the four polyiamonds $\overline{C}_1:=$ o, $\overline{C}_2:=$ w, $\overline{C}_3:=$ s, and $\overline{C}_4:=$ p, i.e., each \Cfree polyiamond is contain{ed} in some polyiamond in \nC. It remains to show that all of these do not fold into \Oct.

	\begin{claim}
		The polyiamond $\overline{C}_1$ does not fold into \Oct.
	\end{claim}
	The polyiamond $\overline{C}_1$ contains a \hex, see \cref{fig:notC1}.  By \cref{lem:hex}(\cref{prop:hex}),  the contained \hex covers at most four faces of the octahedron~\Oct. Hence, in every partial folding of $\overline{C}_1$ into~\Oct,  $\overline{C}_1$  covers at most seven faces of \Oct. Consequently, it does not fold into \Oct.
	
	\begin{figure}[htb]
		\centering
		\begin{subfigure}[t]{.2\textwidth}
			\centering
			\includegraphics[page=14]{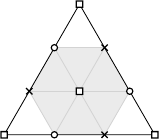}
			\caption{A \hex in $\overline{C}_1$.}
			\label{fig:notC1}
		\end{subfigure}\hfil
		\begin{subfigure}[t]{.25\textwidth}
			\centering
			\includegraphics[page=12]{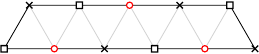}
			\caption{A 3-coloring of $\overline{C}_2$.}
			\label{fig:notC2}
		\end{subfigure}
		\caption{Illustration for the proof that  $\overline{C}_1$  and $\overline{C}_2$ do not fold into $\Oct$.}
		\label{fig:unfoldable2}
	\end{figure}

	\begin{claim}
		The polyiamond $\overline{C}_2$ does not fold into \Oct.
	\end{claim}
	For the purpose of a contradiction, we assume that $\overline{C}_2$ does fold into \Oct. We consider a 3-coloring as illustrated in \cref{fig:notC2} and use the fact that each color class is mapped to a pair of antipodal pairs by \cref{lem:colorclasses}. Note that there exist only three circle vertices, each of which is adjacent to three faces of $\overline{C}_2$. Hence, one (of the two antipodal) corner of \Oct is covered by only one circle vertex implying that not all of its incident faces are covered.

	\begin{claim}\label{clm:unfoldability}
		The polyiamond $\overline{C}_3$ does not fold into \Oct.
	\end{claim}
	The  polyiamond $\overline{C}_3$ can be viewed as copies of \hex and  \dhex overlapping in two triangles, see \cref{fig:unfC3a}. For the purpose of a contradiction,  we consider a 3-coloring of $\overline{C}_3$ as illustrated in \cref{fig:unfC3c}. 
	Note that there are four square vertices in total; we denote them by $v_1,v_2,v_3,v_4$.
	Moreover, the three leftmost square vertices cannot all map to the same corner $c$; otherwise the left \hex maps to at most two faces and $\overline{C}_3$ covers at most $2+6-1=7$ faces.
	We distinguish two cases.

	\begin{figure}[b]
		\centering
		\begin{subfigure}[b]{.33\textwidth}
			\centering
			\includegraphics[page=7]{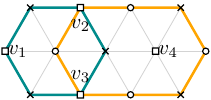}
			\caption{$\overline{C}_3$ consists of a \hex and a \dhex overlapping in two triangles.}
			\label{fig:unfC3a}
		\end{subfigure}\hfil
		\begin{subfigure}[b]{.3\textwidth}
			\centering
			\includegraphics[page=9]{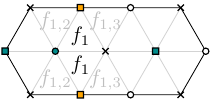}
			\caption{Case: $v_2$ and $v_3$ are mapped to same corner.}
			\label{fig:unfC3c}
		\end{subfigure}\hfil
		\begin{subfigure}[b]{.3\textwidth}
			\centering
			\includegraphics[page=8]{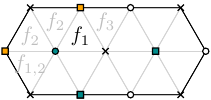}
			\caption{Case: $v_2$ and $v_3$ are mapped to antipodal corners.}
			\label{fig:unfC3b}
		\end{subfigure}
		\caption{Illustration for the non-foldability of $\overline{C}_3$.}
		\label{fig:unfC3}
	\end{figure}

	If $v_2$ and $v_3$ are mapped to $c$ (and $v_1$ to the antipodal corner $\overline c$), then all of their six incident triangles contain one of two neighboring edges of $c$; for an illustration consider  \cref{fig:unfC3c}. Hence, they cover at most three faces incident to $c$. 
	Moreover, $v_4$ must map to $\overline c$; otherwise the two triangles incident to $v_1$ are the only ones mapping to any of the four faces incident to $\overline c$.
	Consequently, all remaining triangles map to a face incident to $\bar c$ and thus, they are not able to cover the remaining face incident to $c$.
A contradiction.

It remains to consider the case that $v_2$ and $v_3$ are mapped to two antipodal corners $c$ and $\overline c$, respectively.
We may assume without loss of generality that $v_1$ is mapped to the corner~$c$  as illustrated in \cref{fig:unfC3b}. Then $v_4$ is mapped to the antipodal $\overline c$; otherwise not all faces of $\overline c$ are covered. Consequently, all triangles incident to $c$ are incident to $v_1$ and $v_2$. However,  four (of the five) triangles incident to $v_1$ and $v_2$ share one edge of \Oct. Hence the five triangles of $v_1$ and $v_2$ cover at most three faces incident to $c$. A contradiction to the foldability of~$\overline{C}_3$.
	

	\begin{claim}\label{prop:unfoldable12}
		The polyiamond $\overline{C}_4$ does not fold into \Oct.
	\end{claim}
	The  polyiamond $\overline{C}_4$ consists of a \dhex and a \hex overlapping in three triangles as illustrated in \cref{fig:unfC4a}. If their intersection is mapped to three different faces of \Oct, then by \cref{lem:hex}(\cref{prop:hex}) and (\cref{prop:dhex}), $\overline{C}_4$ covers at most $4+6-3=7$ faces of \Oct. Consequently, in every folding of  $\overline{C}_4$ into \Oct, the triangles in the considered intersection map to at most two distinct faces. 
	In the following, we focus on the four central triangles of $\overline{C}_4$. By the above observation and the rotational symmetry of $\overline{C}_4$, the triangles of each `line' are mapped to at most two distinct faces.
	We distinguish two cases.

	\begin{figure}[htb]
		\centering
		\begin{subfigure}[t]{.33\textwidth}
			\centering
			\includegraphics[page=3]{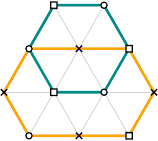}
			\caption{$\overline{C}_4$ consists of a \hex and a \dhex overlapping in three triangles.}
			\label{fig:unfC4a}
		\end{subfigure}
	\hfil
		\begin{subfigure}[t]{.3\textwidth}
			\centering
			\includegraphics[page=4]{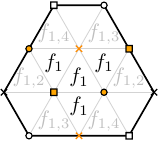}
			\caption{Case: the four central triangles map to the same face.}
			\label{fig:unfC4b}
		\end{subfigure}\hfil
		\begin{subfigure}[t]{.3\textwidth}
			\centering
			\includegraphics[page=5]{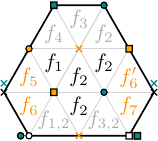}
			\caption{Case: the four central triangles map to two faces.}
			\label{fig:unfC4c}
		\end{subfigure}
		\caption{Illustration for the non-foldability of $\overline{C}_4$.}
		\label{fig:unfC4}
	\end{figure} 

	 If all four triangles map to the same face, denoted by $f_1$, then consider \cref{fig:unfC4b}. By their common edge incident to $f_1$, each of the two triangles with label~$f_{1,i}$ \cref{fig:unfC4b}, $i\in\{2,3,4\}$, cover at most one face different from $f_1$. Consequently,  at most seven faces can be covered in total and this case does not yield a folding of $\overline{C}_4$ into~\Oct.
	
	If the four central triangles map to two different faces, then by \cref{lem:colorclasses}, the map is as illustrated in \cref{fig:unfC4c}. 
	By \cref{lem:hex}(\cref{prop:hex}) and (\cref{prop:dhex}),  the copy of \hex covers at most four faces and the copy of \dhex covers at most six faces. Subtracting the double count of the intersection, the triangles of $\overline{C}_4$ cover at most $4+6-2=8$ faces. 
	Hence, by \cref{lem:hex}(\cref{prop:hex2}), the top copy of \hex covers exactly four faces and is consistent with the triangle-face-map of \cref{fig:C6b}.
	Note that the two triangles with label $f_{i,2}$, $i\in\{1,3\}$, in \cref{fig:unfC4c} contain the common edge of $f_i$ and $f_2$ and thus they may not cover new faces of \Oct. It follows that the remaining four triangles cover distinct and new faces of \Oct. However, this implies that  triangles $f_6$ and $f_6'$ are mapped to the same face. A contradiction. Hence,  $\overline{C}_4$ does not fold into~\Oct.
\end{proof}

\section{A Sharp Size Bound}\label{sec:4}
As shown in \cref{clm:unfoldability}, 
the polyiamond~$\overline{C}_3$
 is not foldable, i.e.,
there exist  polyiamonds of size 14 that do not fold into \Oct. 
In this section, we show the following complementing theorem.

\begin{restatable}{theorem}{thmSize}\label{thm:size}
	Every polyiamond $P$ of size $\geq15$ folds into~\Oct. 
\end{restatable}
To present an idea of the proof, we give some useful sufficient conditions and a simple upper bound.
Let $P$ be a polyiamond and $\ell$ some grid line. The \emph{$\ell$-width} of $P$ denotes the size of the polyiamond obtained by folding all edges parallel to  $\ell$ in a zig-zag-manner as indicated in \cref{fig:convexReduction}. 
The \emph{width} of $P$ is the maximum of the three different $\ell$-widths.  Because the convex  polyiamond $\Pstrich:=z$, depicted in \cref{fig:convex_foldable_constructionTREE}, folds into \Oct, we obtain the following.

\begin{lem}\label{lem:width}
	Every polyiamond $P$ of width at least 10 folds into \Oct.
\end{lem}

\begin{proof}
	Because $P$ has width 10, it can be folded into the polyiamond $\Pstrich$ by zig-zag-folds. Then, by  \cref{thm:convex},  $\Pstrich$ can be folded into $\Oct$.
\end{proof}

Moreover, we determine an upper bound on the size of polyiamonds of width $\leq 9$. In particular, they have size~$\leq 42$ which yields a nice and simple upper bound.

\begin{restatable}{corollary}{corEasyUpperBound}\label{cor:EasyUpperBound}
	Every polyiamond of size $> 42$ folds into~\Oct.
\end{restatable}

\begin{proof}
	Let $P$ be a polyiamond that does not fold into~\Oct.
	Then, by \cref{lem:width}, $P$ has width~$\leq 9$.
	Consequently,  $P$ is contained in the intersection of three strips of width~9 with different rotation. As illustrated in \cref{fig:strip2A,fig:strip2B},
\begin{figure}[bh]
	\centering
	\begin{subfigure}[t]{.47\textwidth}
		\centering
		\includegraphics[page=3,scale=.5]{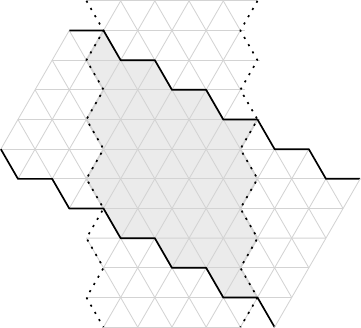}
		\caption{}
		\label{fig:strip2A}
	\end{subfigure}\hfill
	\begin{subfigure}[t]{.47\textwidth}
		\centering
		\includegraphics[page=4,scale=.5]{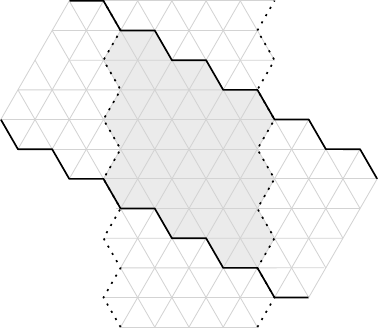}
		\caption{}
		\label{fig:strip2B}
	\end{subfigure}\hfil
	
	\begin{subfigure}[t]{.2\textwidth}
		\centering
		\includegraphics[page=6,scale=.5]{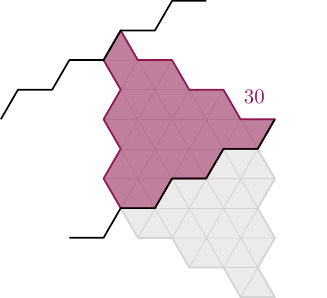}
		\caption{}
	\end{subfigure}\hfill
	\begin{subfigure}[t]{.2\textwidth}
		\centering
		\includegraphics[page=7,scale=.6]{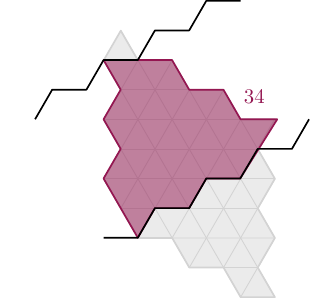}
		\caption{}
	\end{subfigure}\hfill
	\begin{subfigure}[t]{.2\textwidth}
		\centering
		\includegraphics[page=8,scale=.6]{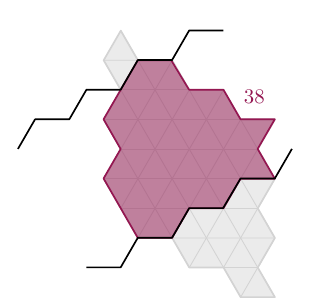}
		\caption{}
	\end{subfigure}\hfill
	\begin{subfigure}[t]{.2\textwidth}
		\centering
		\includegraphics[page=9,scale=.5]{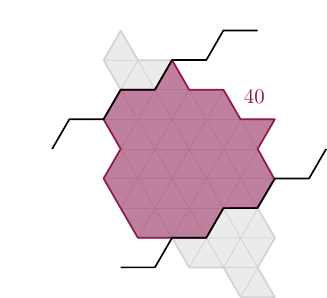}
		\caption{}
	\end{subfigure}\hfill
	\begin{subfigure}[t]{.2\textwidth}
		\centering
		\includegraphics[page=10,scale=.5]{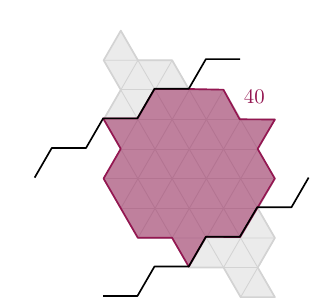}
		\caption{}
	\end{subfigure}\hfill
	\begin{subfigure}[t]{.2\textwidth}
		\centering
		\includegraphics[page=12,scale=.5]{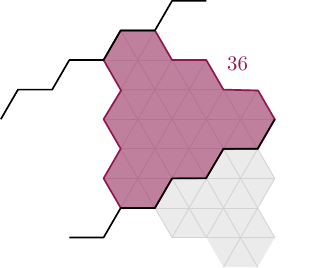}
		\caption{}
	\end{subfigure}\hfill
	\begin{subfigure}[t]{.2\textwidth}
		\centering
		\includegraphics[page=13,scale=.5]{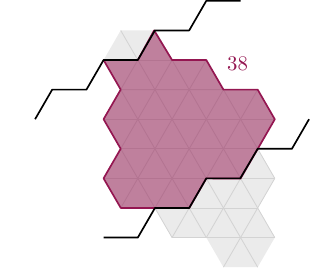}
		\caption{}
	\end{subfigure}\hfill
	\begin{subfigure}[t]{.2\textwidth}
		\centering
		\includegraphics[page=14,scale=.5]{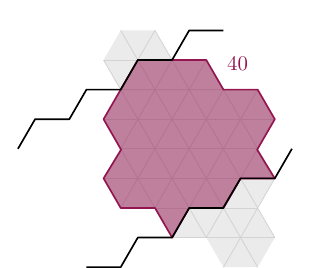}
		\caption{}
	\end{subfigure}\hfill
	\begin{subfigure}[t]{.2\textwidth}
		\centering
		\includegraphics[page=15,scale=.5]{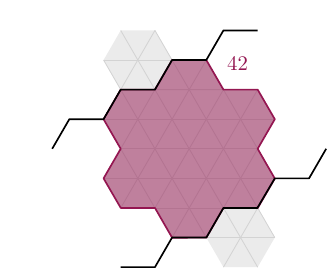}
		\caption{}
		\label{fig:42}
	\end{subfigure}
	\caption{Illustration for the proof of \cref{cor:EasyUpperBound}. Construction of the maximal polyiamonds of width $\leq 9$; their sizes are indicated by numbers.}
	\label{fig_regionsNEW}
\end{figure}
 two of these infinite strips may intersect in two distinct ways. The intersection with \new{different translations of}  a third strip results in nine polyiamonds (six of which are pairwise different), see \cref{fig_regionsNEW}.

The largest of these polyiamonds has size~42 and is depicted in \cref{fig:42}. Because $P$ is contained in one of them, $P$ 
has size at most 42. Consequently, any larger polyiamond is foldable.
\end{proof}

To show the sharp bound, we need to work a little harder.
In particular, the proof is computer-aided.

\paragraph{Proof of the sharp upper bound}
\cref{thm:size} is based on a strong sufficient criterion. 
Let \Px, \Pu,  \Pz, and \Pl denote the polyiamonds depicted in \cref{fig:Px,fig:Pu,fig:Pz,fig:Pl}, respectively. 
In a first step, we show that polyiamonds that are large enough and do contain one of the four polyiamonds fold into $\Oct$.

\begin{figure}[htb]
	\centering
	
	\begin{subfigure}[t]{.2\textwidth}
		\centering
		\includegraphics[page=19,scale=\scaleL]{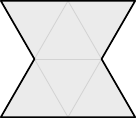}
		\caption{\Px}
		\label{fig:Px}
	\end{subfigure}\hfill
	\begin{subfigure}[t]{.2\textwidth}
		\centering
		\includegraphics[page=24,scale=\scaleL]{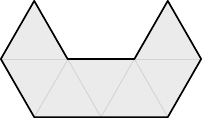}
		\caption{\Pu}
		\label{fig:Pu}
	\end{subfigure}\hfill
	\begin{subfigure}[t]{.2\textwidth}
		\centering
		\includegraphics[page=26,scale=\scaleL]{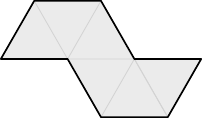}
		\caption{\Pz}
		\label{fig:Pz}
	\end{subfigure}\hfill
	\begin{subfigure}[t]{.2\textwidth}
		\centering
		\includegraphics[page=25,scale=\scaleL]{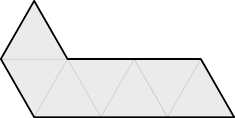}
		\caption{\Pl}
		\label{fig:Pl}
	\end{subfigure}
	\caption{Illustration of the four polyiamonds used in \cref{prop:containingP}.}
	\label{fig:suffienct}
\end{figure}

\begin{lem}\label{prop:containingP}
	Every polyiamond $P$ that \Tcontains \Px, \Pu, \Pz, or \Pl and has size~$\geq 15$ folds into~\Oct.
\end{lem}
Before proving \cref{prop:containingP}, we show how to deduce \cref{thm:size}.
\thmSize*
\begin{proof}
	We call a polyiamond \emph{\Pfree} if it does not  \Tcontain any of the  polyiamonds \Pstrich, \Px, \Pu, \Pz, or \Pl. 
	By \cref{thm:convex,prop:containingP}, it remains to show that no \Pfree polyiamond of size $\geq15$ exists. To do so, we construct all \Pfree polyiamonds bottom-up and show that indeed there exists no such polyiamond. The construction is similar to the one in the proof of \cref{thm:convex}: We start with the polyiamond of size 1. Then, we enlarge every \Pfree polyiamond of size $k$ by individual triangles and check if the resulting polyiamonds remain \Pfree. In this way, we obtain a list of \Pfree polyiamonds of size $k+1$.
	\cref{table:1} 
	presents the numbers $p(n)$ of \Pfree  polyiamonds with  size $n$; these numbers have been generated by computer-search. The code is available at \url{https://github.com/dasnessie/folding-polyiamonds/}.
\end{proof}
\begin{table}[hb]
	\centering
	\caption{The number $p(n)$ of \Pfree polyiamonds of size~$n$. }
	\label{table:1}
	\begin{tabular}{c| cc cccc c  cc cccc c}
		$n$    & 2&$3$ & 4 & 5 &  6 &  7 & 8  &  9 & 10 & 11 & 12 & 13  &  14& 15\\ \hline \vspace{6pt}
		$p(n)$ & 1 & 1        & 3 & 4 & 10 & 16 & 22 & 22 & 16 &  9  & 3 &  1  & 0  & 0 
	\end{tabular}
\end{table}

It remains to prove \cref{prop:containingP}. We split its proof into four claims. Together, \cref{lem:containingPx,lem:containingPu,lem:containingPz,lem:containingPl} imply \cref{prop:containingP}.

\begin{claim}\label{lem:containingPx}
	Every polyiamond $P$ that \Tcontains \Px and has size $\geq 15$ folds into~\Oct.
\end{claim}
\begin{proof}
	We consider the \dhex-frame containing \Px as illustrated in \cref{fig:unfC10plus5B} and call each connected group of rose triangles a \emph{flap} of \Px.
	We consider the following cases:
	
	If triangles exist in two distinct flaps, then there exists a triangle-face-map such that some triangles are mapped to (the two missing faces) $f_7$ and $f_8$, see \cref{fig:unfC10plus5B} (or its mirror image). First, we fold away all (but at most two) triangles that are not contained in the two flaps. The two corner triangles between two flaps may remain. Its foldability is implied by the fact that the polyiamond in \cref{fig:unfC10plus5B} folds into~\Oct.

	\begin{figure}[htb]
		\centering
		\begin{subfigure}[t]{.3\textwidth}
			\centering
			\includegraphics[page=13]{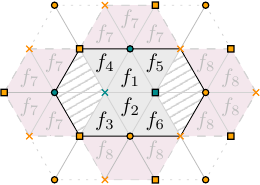}
			\caption{two flaps}
			\label{fig:unfC10plus5B}
		\end{subfigure}\hfill
		\begin{subfigure}[t]{.3\textwidth}
			\centering		
			\includegraphics[page=15]{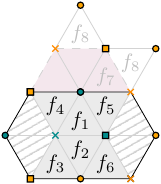}
			\caption{top or bottom flap}
			\label{fig:unfC10plus5D}
		\end{subfigure}\hfill
		\begin{subfigure}[t]{.3\textwidth}
			\centering
			\includegraphics[page=14]{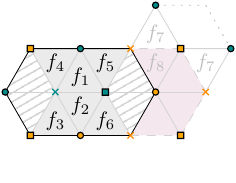}
			\caption{left or right flap}
			\label{fig:unfC10plus5C}
		\end{subfigure}
		\caption{Illustration of the proof of \cref{lem:containingPx}.}
		\label{fig:unfC10plus5}
	\end{figure}
	
	It remains to consider the case that $P$ without $\dhex$ is attached via only one flap. Because $P$ has size at least 15 and each flap has size at most 4, there exists a triangle outside the flap. \cref{fig:unfC10plus5D,fig:unfC10plus5C} shows that this guarantees foldability in all cases.
\end{proof}

\begin{claim}\label{lem:containingPu}
	Every polyiamond $P$ that \Tcontains \Pu and has size $\geq 15$ folds into~\Oct.
\end{claim}

\begin{proof}
	We may assume that $P$ does not \Tcontain $\dhex$ (which \Tcontains $\Px$); otherwise \cref{lem:containingPx} implies the \new{statement}. Consequently, $P$ has six triangles outside the  $\dhex$-frame containing $P$, see \cref{fig:something}.
	We distinguish two cases:  a) there exists a triangle in some flap with a neighboring triangle outside the flap or b) all triangles of $P$ lie within the flaps.

	\begin{figure}[htb]
		\centering		
		\begin{subfigure}[t]{.45\textwidth}
			\centering
			\includegraphics[page=22]{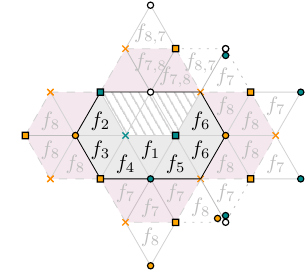}
			\caption{}
			\label{fig:something}
		\end{subfigure}\hfil
		\begin{subfigure}[t]{.45\textwidth}
			\centering
			\includegraphics[page=23]{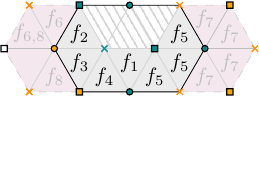}
			\caption{}
			\label{fig:something2}
		\end{subfigure}
		\caption{Illustration of the proof of \cref{lem:containingPu}.}
		\label{fig:}
	\end{figure}

		In case a),  the map in \cref{fig:something} can be reflected (horizontally) such that there exist triangles with labels $f_7$ and $f_8$. The label $f_{7,8}$ indicates that it can be adjusted as wished. Moreover, $P$ folds into~\Oct by some strategy presented in \cref{fig:something} (after reducing to a crucial convex subpolyiamond).
	
	In case b), unless all triangles lie within the left and right flap, the map in \cref{fig:something} can be reflected  such that there exist faces with labels $f_7$ and $f_8$. Moreover, $P$ folds into~\Oct by the strategy presented in \cref{fig:something}.
	
	If all triangles lie within the left and right flap, we consider the strategy indicated in \cref{fig:something2}. By symmetry, we may assume that the left flap contains at least three triangles. Hence, there exist faces with label $f_6$ and $f_8$; moreover, a triangle with label $f_7$ exists in the right flap. As the polyiamond in \cref{fig:something2} folds into~\Oct, $P$ does as well.
	This completes the proof.
\end{proof}

\begin{claim}\label{lem:containingPz}
	Every polyiamond $P$ that \Tcontains \Pz and has size $\geq 15$ folds into~\Oct.
\end{claim}
\begin{proof}
	We may assume that $P$ does not \Tcontain \Px nor \Pu; otherwise \cref{lem:containingPx,lem:containingPu} imply the \new{statement}. Consequently, $P$ has seven triangles outside the \dhex-frame depicted in \cref{fig:ProofPzA}.

\begin{figure}[htb]
	\centering		
	\begin{subfigure}[t]{.45\textwidth}
		\centering
		\includegraphics[page=29]{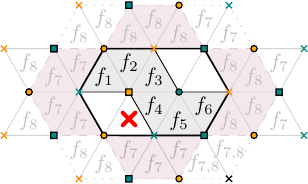}
		\caption{}
		\label{fig:ProofPzA}
	\end{subfigure}\hfil
	\begin{subfigure}[t]{.45\textwidth}
		\centering
		\includegraphics[page=30]{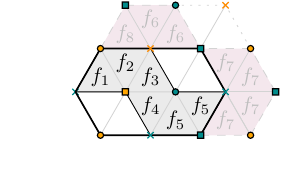}
		\caption{}
		\label{fig:ProofPzB}
	\end{subfigure}
	\caption{Illustration for the proof of \cref{lem:containingPz}.}
	\label{fig:ProofPz}
\end{figure}
	Similar as above, we consider the case that there exists a flap with a neighboring triangle outside the flap. Then,  the map in \cref{fig:ProofPzA} induces triangles with labels $f_7$ and~$f_8$. Moreover, the depicted polyiamonds folds into \Oct and \Tcontains~$P$. 
The same argument can be applied for the case that there exist triangles in flaps with different labels.

It remains to consider the case that all triangles are contained in neighboring flaps with the same labels. By the rotational symmetry, we may assume that all triangles are contained in the top and right flap as illustrated in \cref{fig:ProofPzB}; moreover, we know that all of these triangles are present because at least seven triangles exist outside the \dhex-frame. The illustrated polyiamond folds into \Oct and \Tcontains~$P$. Thus, $P$ folds into \Oct.
\end{proof}

\begin{claim}\label{lem:containingPl}
	Every polyiamond $P$ that \Tcontains \Pl and has size $\geq 15$ folds into~\Oct.
\end{claim}
\begin{proof}
	Observe that every triangle outside the dashed frame in \cref{fig:ProofPlA} yields a triangle with the missing label $f_8$. Hence, we may assume that $P$ is contained in the frame.
	Moreover, we may assume that $P$ does not contain \Pu nor \Pz; otherwise \cref{lem:containingPu,lem:containingPz} imply the \new{statement}. Consequently, at least 3 triangles are missing within the frame as indicated, where crosses on an edge indicate that at most one of the incident triangles exist.

	\begin{figure}[htb]
		\centering	
		\begin{subfigure}[t]{.35\textwidth}
			\centering
			\includegraphics[page=31]{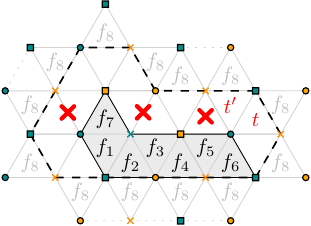}
			\caption{}
			\label{fig:ProofPlA}
		\end{subfigure}\hfil
		\begin{subfigure}[t]{.3\textwidth}
			\centering
			\includegraphics[page=32]{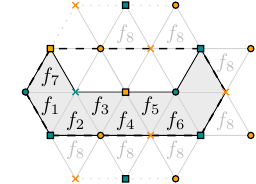}
			\caption{}
			\label{fig:ProofPlB}
		\end{subfigure}\hfil
		\begin{subfigure}[t]{.3\textwidth}
			\centering
			\includegraphics[page=33]{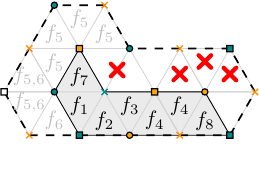}
			\caption{}
			\label{fig:ProofPlC}
		\end{subfigure}
		\caption{Illustration for the proof of \cref{lem:containingPl}.}
		\label{fig:ProofPl}
	\end{figure}

	We distinguish two cases:
If $P$ contains the triangle $t$, then it contains the polyiamond depicted in \cref{fig:ProofPlB}. Because the frame contains only 14 triangles, there exists a triangle $f$ outside the frame. Together with the depicted map (or its mirror image), $f$ ensures a triangle with label $f_8$. 
	
	If $P$ does not contain the triangle $t$, then its left neighboring triangle $t'$ does not belong to $P$ because $P$ is contained in the dashed frame, see \cref{fig:ProofPlC}.
	The frame-polyiamond depicted in \cref{fig:ProofPlC} has four triangles with label $f_5$, one with $f_6$ and two with joker label $f_{5,6}$ (indicating that these can be adjusted as wished).
	Because $P$ does not \Tcontain \Pz nor \Pu, at most three of the remaining triangles without labels exist. It follows that  any choice of eight additional triangles contains two triangles with different labels. By adjusting the joker label as needed, these labels can represent $f_5$ and $f_6$. Moreover, the depicted polyiamonds fold into~\Oct.
\end{proof}
%


%

\section{Folding with Prescribed Face Coverage}\label{sec:5}
For a polyiamond folded into the octahedron, we call the number of triangles covering each face its  \emph{\cover}. Clearly, the \cover is a positive integer; otherwise the polyiamond does not fold into \Oct.
In this section, we aim to find polyiamonds that fold into the octahedron such that each face~$f_i$ of the octahedron is covered by an 
assigned positive integer $m_i$ of triangles. 
We show that for each choice of assigned positive integers, there exists a foldable polyiamond such that the \cover of each face is equal to the assigned number.

\begin{theorem}\label{thm:multplicity}
	Let $(m_i)_{i=[8]}$ be a sequence of positive integers. Then there exists a polyiamond $P$ that folds into~\Oct such that the \cover of face $f_i$ is $m_i$. 
	
	In particular, there exists such a polyiamond $P$ that does neither contain holes nor slit edges.
\end{theorem}

\begin{proof}
	For simplicity, we start by presenting a polyiamond with slit edges that folds into~\Oct as required. Afterwards, we modify the strategy to obtain polyiamonds without slit edges.
	
	\cref{fig:MultiplicitySlits} displays a polyiamond that can be viewed as an octahedral net with \emph{arms} of appropriate length. Apart from the net, each face$f_i$ is assigned an arm~$A_i$ consisting of $m_i-1$ triangles. Each arm can be folded on the incident triangle of the net by folding the grid edges using alternating mountain and valley folds. 
	
	\begin{figure}[htb]
		\centering
		\begin{subfigure}[t]{.45\textwidth}
			\centering
			\includegraphics[page=15]{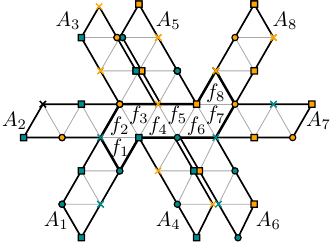}
			\caption{A polyiamond with slit edges.}
			\label{fig:MultiplicitySlitsA}
			\end{subfigure}\hfil
		\begin{subfigure}[t]{.45\textwidth}
			\centering
			\includegraphics[page=16]{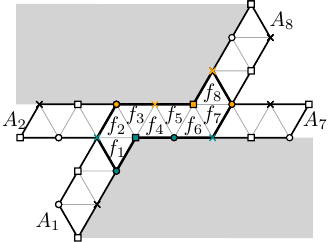}
			\caption{Idea for polyiamond without slit edges.}
			\label{fig:MultiplicitySlitsB}
		\end{subfigure}\hfil
	\caption{Illustration for the proof of \cref{thm:multplicity}.}
	\label{fig:MultiplicitySlits}
	\end{figure}
	
	In order to obtain a polyiamond without slit edges, a little more work is required. 
	The arms $A_1,A_2,A_7,A_8$ remain unchanged as depicted in \cref{fig:MultiplicitySlitsB}. The idea is to attach the arm $A_4$ to the arm $A_3$ such that both arms stay within the top gray region. Symmetrically, the arm $A_5$ is attached to $A_6$.  We distinguish a few simple cases:

	If $m_3\geq 4$, then we attach the arm $A_4$ on $A_3$ as depicted in \cref{fig:MultiplicityA}. 
	The coloring of the vertices shows that the triangle of each arm can be folded onto the correct face of~\Oct.
		\begin{figure}[htb]
		\centering
		\begin{subfigure}[t]{.45\textwidth}
			\centering
			\includegraphics[page=18]{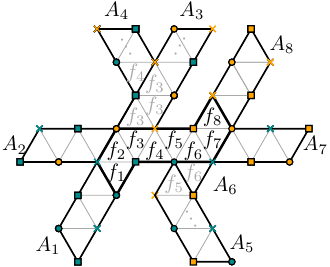}
			\caption{Case: $m_3\geq 4$ or $m_3=2$  }
			\label{fig:MultiplicityA}
		\end{subfigure}\hfil
		\begin{subfigure}[t]{.45\textwidth}
			\centering
			\includegraphics[page=17]{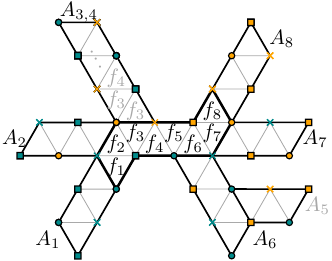}
			\caption{Case: $m_3=3$.}
			\label{fig:MultiplicityB}
		\end{subfigure}\hfil
		\caption{Illustration for the proof of \cref{thm:multplicity}.}
		\label{fig:Multiplicity}
	\end{figure}

	Similarly, if $m_3=2$, we attach the arm $A_4$ on the right side of the arm $A_3$ consisting of one triangle as in \cref{fig:MultiplicityA} (illustrated for $A_5$ and $A_6$).
	
	If $m_3=3$, then we consider the polyiamond in \cref{fig:MultiplicityB} with arm $A_{3,4}$ of length $(m_3-1)+(m_4-1)$. Note that we may fold the first two triangles  onto $f_3$ and the remaining ones onto $f_4$.
	
	By symmetry, we handle the cases of $m_6=2,3,\geq4$ analogously.
	Hence, it remains to consider the case that  $m_3=1$ or $m_6=1$ (or both). Note that $f_3$ and $f_6$ do not share vertices and hence, they are opposite faces. Using the symmetry of the octahedron, we may assume that $f_3,f_6$ is a pair of opposite faces maximizing $\min\{m_3,m_6\}$. It thus remains to consider the case that all four pairs of opposite faces contain a face with \cover 1. Either, two of these faces share an edge or there exist exactly four faces with $m_i=1$, none of which share an edge. 
	If two of the faces with \cover 1 share a side, we map them to $f_4$ and $f_5$. Note that removing the arms $A_4,A_5$ in \cref{fig:MultiplicitySlitsA} yields a polyiamond without slit edges.
	Otherwise, we may assume that $m_1=m_3=m_5=m_7=1$ and we consider the polyiamond in \cref{fig:MultiplicityC}. It consists of a net and four arms that can be folded onto its respective faces as before.
	\begin{figure}[htb]
		\centering
		\includegraphics[page=19]{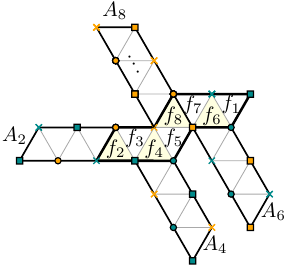}
		\caption{Illustration for the proof of \cref{thm:multplicity} for the case that $m_1=m_3=m_5=m_7=1$.}
		\label{fig:MultiplicityC}
		\end{figure}
\end{proof}

\section{Future work}\label{sec:6}

In this paper, we studied foldability of polyiamonds into the octahedron. For future work, it is interesting if similar results can be obtained for foldings into other platonic solids such as the dodecahedron and the icosahedron. We note that the number of nets increases rapidly: While the cube and the octahedron have 11 nets, the dodecahedron and icoshahedron have 43\,380 nets \cite{PlatonicSolidsEdgeUnfoldings}.

\paragraph{Acknowledgments}
We thank Christian Rieck for valuable suggestions on a draft of this manuscript and the anonymous reviewers for constructive feedback.
Additionally, the first author thanks Tilman Stehr for his good advice concerning questions of implementation. 


\small
\newcommand*{\doi}[1]{\href{http://dx.doi.org/#1}{doi: #1}}
\bibliographystyle{abbrvurl}
\bibliography{bibliography}

\end{document}